\documentclass[english,autoref,nolineno]{socg-lipics-v2019}
\usepackage{cite,microtype,graphicx}

\title{Cubic Planar Graphs That Cannot Be Drawn On Few Lines}

\author{David Eppstein}{Computer Science Department, University of California, Irvine, USA}{eppstein@uci.edu}{}{Supported in part by NSF grants  CCF-1618301 and CCF-1616248.}
\authorrunning{David Eppstein}
\Copyright{David Eppstein}

\EventEditors{Gill Barequet and Yusu Wang}
\EventNoEds{2}
\EventLongTitle{35th International Symposium on Computational Geometry (SoCG 2019)}
\EventShortTitle{SoCG 2019}
\EventAcronym{SoCG}
\EventYear{2019}
\EventDate{June 18--21, 2019}
\EventLocation{Portland, United~States}
\EventLogo{socg-logo}
\SeriesVolume{129}
\ArticleNo{32} % <- Your article number goes here

\begin{document}
\maketitle  

\begin{abstract}
For every integer $\ell$, we construct a cubic 3-vertex-connected planar bipartite graph $G$
with $O(\ell^3)$ vertices such that there is no planar straight-line drawing of $G$ whose vertices all lie on $\ell$ lines. This strengthens previous results on graphs that cannot be drawn on few lines, which constructed significantly larger maximal planar graphs. We also find apex-trees and cubic bipartite series-parallel graphs that cannot be drawn on a bounded number of lines.
\end{abstract}

\section{Introduction}

A number of works in graph drawing and network visualization have considered drawing graphs with line segments as edges and with the vertices placed on few lines, or on a minimal number of lines. Even very strong constraints, such as restricting the vertices of a drawing to only two lines, allow many graphs to be drawn~\cite{FirLipStr-GD-18}: every \emph{weakly leveled}  graph drawing (a planar drawing on any number of parallel lines with every edge connecting two vertices on the same or adjacent lines) can be converted into a  drawing on two crossing lines that spirals around the crossing. This conversion allows, for instance, all trees, all outerplanar graphs,  all Halin graphs, all squaregraphs (graphs in which all bounded faces have exactly four sides and all vertices not on the unbounded face have at least four neighbors) and all grid graphs (\autoref{fig:2line-grid}) to be drawn on two lines~\cite{BanDevDuj-Algo-18,FelLioWis-JGAA-03,BanCheEpp-SIDMA-10}.

\begin{figure}[t]
\centering\includegraphics[width=0.5\textwidth]{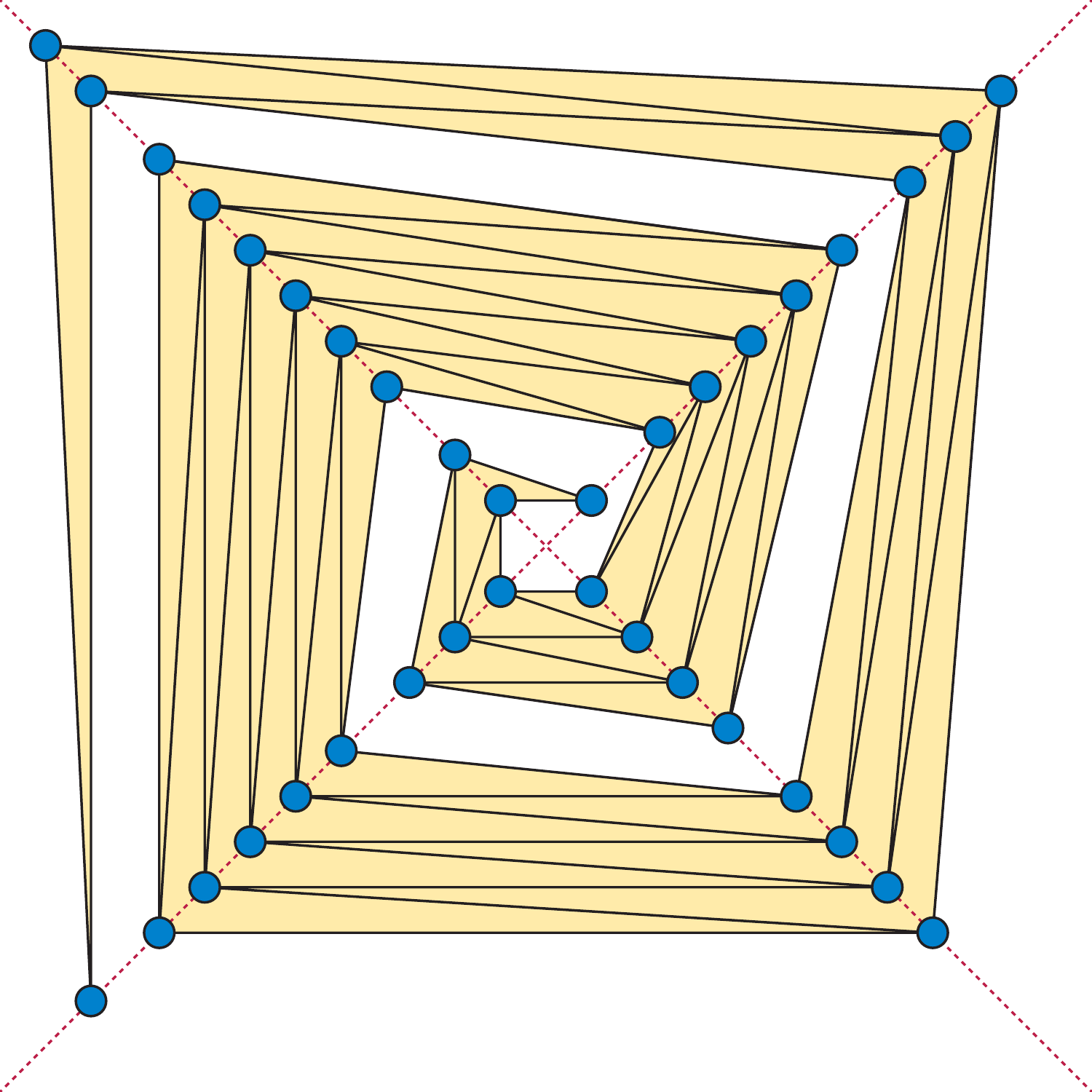}
\caption{Spiraling around the central crossing allows every weakly leveled drawing (in this case, a drawing of a grid graph) to be converted to a drawing with all vertices on two lines.}
\label{fig:2line-grid}
\end{figure}

Additional past results in this area include:
\begin{itemize}
\item Fixed-parameter tractable algorithms for drawing planar graphs without crossings with
all vertices on $\ell$ parallel lines, based on the fact that a graph with such a drawing must have pathwidth $O(\ell)$~\cite{DujFelKit-Algo-08}.
\item NP-hardness of crossing minimization for graphs drawn with vertices on two parallel lines (with a fixed assignment of vertices to lines but variable placement of each vertex along each line) and of finding large crossing-free subgraphs~\cite{EadWhi-TCS-94,EadWor-Algo-94}.
\item NP-hardness of recognizing the graphs that can be drawn without crossing with all vertices on three parallel but non-coplanar lines in three-dimensional space, or on three rays in the plane with a common apex and bounding three wedges with angles less than~$\pi$~\cite{BanDevDuj-Algo-18}.
More generally, the number of parallel three-dimensional lines (in sufficiently general position) needed for a crossing-free drawing of a graph is the \emph{track number} of the graph~\cite{DujPorWoo-DMTCS-04,DujWoo-DMTCS-05}. It is closely related to the volume of three-dimensional grid drawings, and can be bounded by the pathwidth of the graph~\cite{DujMorWoo-GD-02}.
\item $\exists\mathbb{R}$-completeness and fixed-parameter tractability of deciding whether a given graph can be drawn without crossing with all edges on $\ell$ lines (not required to be parallel) in two or three dimensions~\cite{ChaFleLip-WADS-17}.
\item Implementation of a tester for drawing graphs without crossings on two lines using integer linear programming and SAT solvers, and an examination of the subclasses of planar graphs that can be drawn without crossings with all vertices on two lines~\cite{FirLipStr-GD-18}.
\item The existence of families of planar graphs that cannot be drawn without crossings on any fixed number of lines, no matter how the lines are arranged in the plane~\cite{RavVer-WG-11,ChaFleLip-GD-16,Epp-18}.
\end{itemize}

In this paper we strengthen the final result of this listing, the existence of families of planar graphs that cannot be drawn on any fixed number of lines, in two ways.

First, we greatly improve the size bounds for these difficult-to-draw graphs.
The previous bounds of Ravsky, Chaplick, et al.~\cite{RavVer-WG-11,ChaFleLip-GD-16} are based on the observation that, in a maximal planar graph, a line through $q$ vertices implies the existence of a path in the dual graph of length $\Omega(q)$. However, there exist $n$-vertex maximal planar graphs for which the longest dual path has length $O(n^c)$ for some constant $c<1$ called the \emph{shortness exponent}.
In these graphs, at most $O(n^c)$ vertices can lie on one line, so the number of lines needed to cover the vertices of any drawing of such a graph is $\Omega(n^{1-c})$. Based on this reasoning, they showed the existence of $n$-vertex graphs requiring $O(n^{0.01})$ lines to cover the vertices of any drawing.
Inverting this relationship, graphs that cannot be drawn on $\ell$ lines can have a  number of vertices that is only polynomial in $\ell$, but that polynomial is roughly $\ell^{100}$.
Alternatively, it can be proven by a straightforward induction that a special class of maximal planar graphs, the planar 3-trees, cannot be drawn on a constant number of lines, but the proof only shows that the required number of lines for these graphs is at least logarithmic~\cite{Epp-18}. Inverting this relationship, the graphs of this type that cannot be drawn on $\ell$ lines have size exponential in~$\ell$.
In this paper, we prove polynomial bounds with a much smaller exponent than $100$.

\begin{figure}[t]
\centering\includegraphics[width=0.35\textwidth]{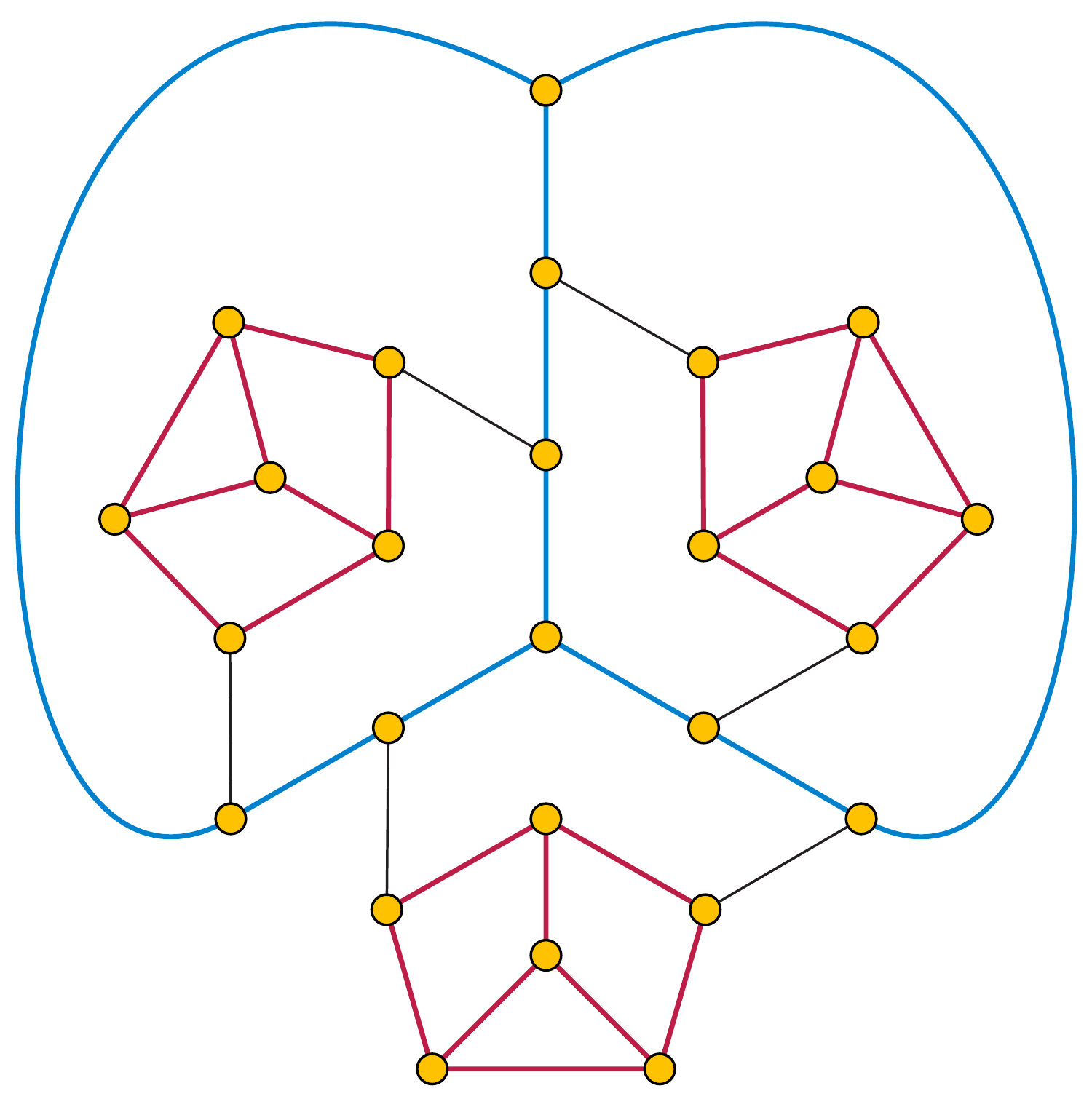}
\caption{A cubic planar graph that cannot be drawn on two lines,
giving a counterexample to a conjecture of Firman et al.~\cite{FirLipStr-GD-18}. Its inability to be drawn on two lines has been verified computationally by Firman et al.}
\label{fig:2line-counterexample}
\end{figure}

Second, we show that the property of requiring many lines extends to a broader class of graphs.
Both classes of counterexamples discussed above involve maximal planar graphs, graphs in which every face of the embedding is a triangle. Maximality seems a necessary part of the proofs based on the shortness exponent, as the connection between numbers of vertices on a single line and lengths of dual paths does not necessarily hold for other classes of planar graphs. In contrast, based on computational experiments, Firman et al.~\cite{FirLipStr-GD-18} conjectured that cubic (that is, 3-regular) planar graphs can always be drawn without crossings on only two lines.  The cubic graphs are distinct from the known examples of graphs that require their vertices to lie on many lines, as a maximal planar graph larger than $K_4$ cannot be cubic.  Their conjecture inspired the present work, and a counterexample to it  found by the author (\autoref{fig:2line-counterexample}) led to our main results. In this work, we provide examples of graphs that require many lines but are cubic, providing stronger counterexamples to the conjecture of Firman et al. Moreover these graphs do not contain any triangles, showing that the presence of triangles is not a necessary component of graphs that require many lines.

More specifically, we prove:

\begin{theorem}
\label{thm:cubic}
For every $\ell$ there exists a graph $G_\ell$ that is cubic, 3-vertex-connected, planar, and bipartite, with $O(\ell^3)$ vertices, such that  every straight-line planar drawing of $G_\ell$ requires more than $\ell$ lines to cover all vertices of the drawing.
\end{theorem}

Additionally, we can construct $G_\ell$ in such a way that it is drawable with its vertices on $O(\ell)$ lines, or such that it has bounded pathwidth. 
In particular, this proves that the relation between pathwidth and number of lines proven by Dujmovi{\'c} et al.~\cite{DujFelKit-Algo-08} is not bidirectional. Every graph that can be drawn without crossings on a given number of parallel lines has  pathwidth bounded linearly in the number of lines, but
the number of lines needed to draw a planar graph cannot be bounded by a function of its pathwidth, regardless of whether we constrain the lines to be parallel.

Using similar methods, we also prove:

\begin{theorem}
\label{thm:treelike}
For every $\ell$ there exists a subcubic series-parallel graph that cannot be drawn with its vertices on $\ell$ lines, and an apex-tree (a graph formed by adding one vertex to a tree)  that cannot be drawn with its vertices on $\ell$ lines.
\end{theorem}

\autoref{thm:treelike} stands in contrast to the fact that all trees and all outerplanar graphs can be drawn on only two crossed lines.
In \autoref{thm:treelike}, both the series-parallel graph and apex-tree can be made to be bipartite, and the apex-tree can be made subcubic at all tree vertices.

These results lower the treewidth of the known graphs that cannot be drawn on a constant number of lines from three~\cite{Epp-18} to two, and they show that adding one vertex to a graph can change the number of lines needed to draw it from two to any larger number. The apex-tree graphs are also central to recent research characterizing the minor-closed families of bounded layered pathwidth~\cite{DujEppJor-ms-18}. However, for \autoref{thm:treelike} we do not have a polynomial size bound on the graphs that we construct; instead, they are exponential in size.
At least in the case of apex-trees, this exponential size blowup is necessary, as we finally prove:

\begin{theorem}
\label{thm:apex-tree-draw}
Every apex-tree with $n$ vertices has a planar embedding that can be drawn with its vertices on $O(\log n)$ parallel lines.
\end{theorem}

\section{Counterexamples of cubic size}

\begin{figure}[t]
\centering\includegraphics[width=0.9\textwidth]{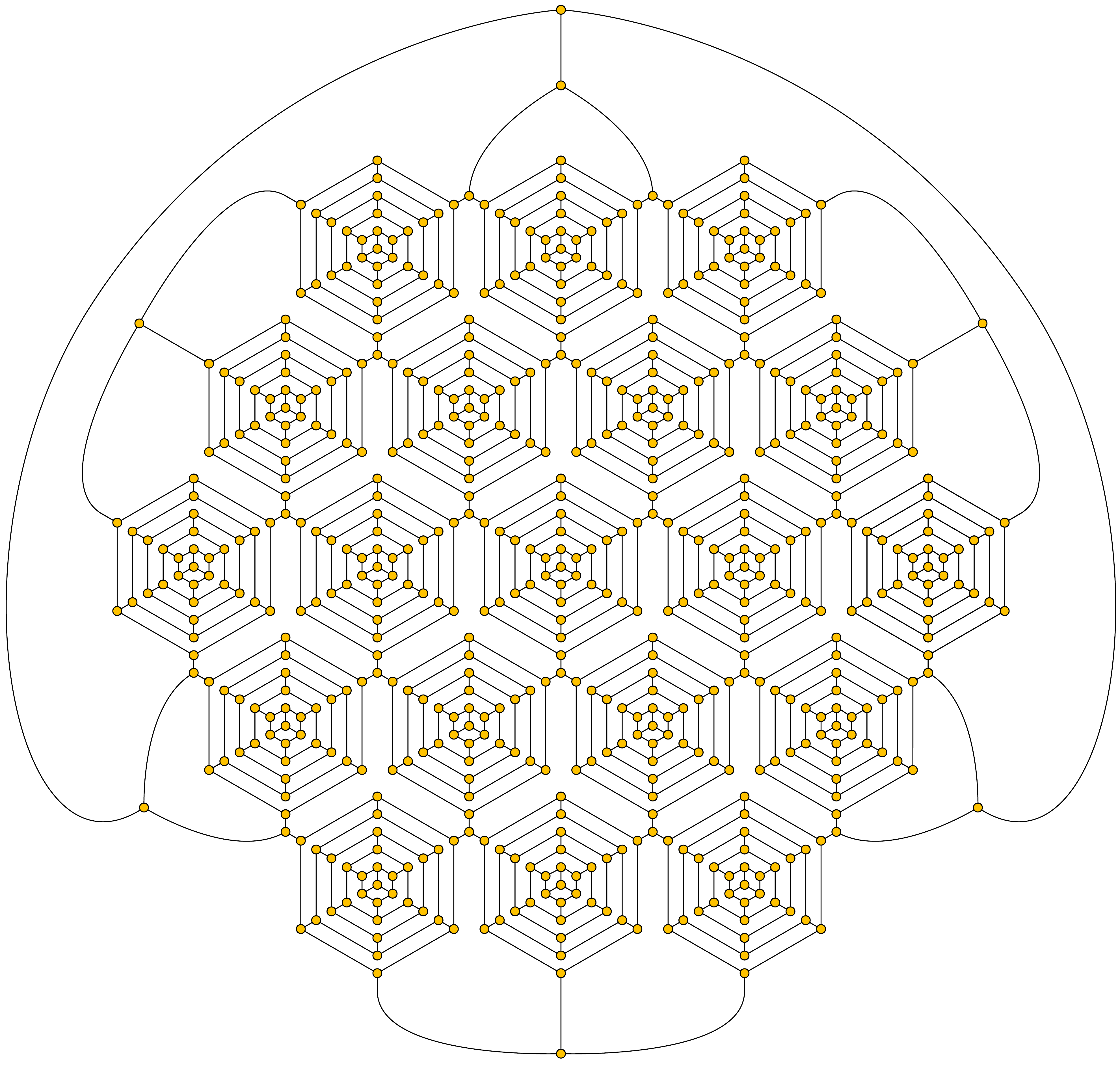}
\caption{The overall construction plan of our graph: $O(\ell^2)$ copies of a subgraph formed from $O(\ell)$ nested hexagons.}
\label{fig:spiderwebs}
\end{figure}

\subsection{Overview}
The overall strategy of our construction for graphs of size $O(\ell^3)$ that cannot be drawn on $\ell$ lines is illustrated in \autoref{fig:spiderwebs}. As can be seen in the figure, the graph consists of a number of subunits, each formed by a set of nested hexagons, connected to each other by triples of edges. These subunits resemble the \emph{nested triangles graph} frequently used as a counterexample in graph drawing~\cite{DolLeiTri-ACR-84,FraPat-GD-07}, but are based on hexagons rather than triangles.

The figure shows faces of three types: triples of quadrilaterals at the center of each set of nested hexagons, non-convex L-shaped hexagonal faces between pairs of hexagons in each set of nested hexagons, and Y-shaped dodecagonal faces between triples of subunits. Because the graph is planar and each bounded face has an even number of sides, the graph is also bipartite. It is straightforward to add a small number of additional vertices surrounding the sets of nested hexagons (as shown) to complete the graph to one that is cubic, 3-vertex-connected, and still bipartite. Because it is 3-vertex-connected, it has a unique planar embedding, the one shown, up to the choice of which face of the embedding is the outer face. With at most one exception (the subunit that includes the outer face), all subunits must be drawn as shown in the figure (topologically but not geometrically), within disjoint hexagonal regions of the plane.

The drawing also shows that each subunit can be drawn using only three lines. Indeed, the number of lines needed to cover all of the vertices in \autoref{fig:spiderwebs} (and in similar figures with more subunits) is proportional only to the square root of the number of subunits. Nevertheless, we shall show that, if there are enough subunits relative to the number of lines, then at least one of the subunits will be difficult to draw on few lines. More specifically, for a given parameter $\ell$ (a number of lines that should be too small to cover all vertices of the drawing) we will choose the number of subunits to be at least $\tbinom{\ell}{2}+2$, two more than the largest possible number of crossing points that $\ell$ lines can have. In this way, there will be at least one subunit that does not include the outer face of the embedding, and does not surround any crossing points of the lines.

We will show that, for subunits formed from a sufficiently large number of nested hexagons (linear in $\ell$), it is not possible to draw the subunit on $\ell$ lines without surrounding any crossing points.
Because, nevertheless, one of the subunits must fail to surround any crossing points, it cannot be drawn on $\ell$ lines. It follows that the whole graph also cannot be drawn on $\ell$ lines.

\subsection{Nested polygons with no surrounded crossings}

To formalize the subunits of the drawing of \autoref{fig:spiderwebs}, we define a \emph{$(p,r)$-nest},
for positive integers $p$ and $r$, to be a collection of $r$ disjoint simple $p$-gons (not necessarily convex), together with one additional point (the \emph{egg}), such that each $p$-gon contains the egg.
Because the $p$-gons are disjoint and all contain the egg, it is necessarily the case (by the Jordan curve theorem) that each two of the $p$-gons are nested, one inside the other. Then the subunits of \autoref{fig:spiderwebs} form a $(6,r)$-nest, for some $r$, together with some additional graph edges between consecutive cycles that force the cycles to be nested within each other but play no additional role in our analysis.

In \autoref{fig:spiderwebs}, each $(6,r)$-nest is drawn with all hexagon vertices on three lines that cross at the egg of the nest. More generally, whenever $p$ is even, a $(p,r)$-nest can be drawn on only two crossing lines, with the egg at the crossing point; when $p$ is odd, three lines suffice.
However, in all of these drawings, a crossing point of the lines is contained within at least one polygon of the nest. We will show that, when this does not happen, nests require $\Omega(r/p)$ lines to cover all of their points.

\begin{lemma}
Let $p$ be a positive integer, let $P$ be a simple $p$-gon, and let $L$ be a line. Then $L$ intersects the interior of $P$ in at most $\lfloor p/2\rfloor$ open line segments.
\end{lemma}

\begin{proof}
Each line segment begins and ends at points where $L$ intersects $P$ either at a vertex or at an interior point of one of the sides of $P$. If the segment endpoint is at a crossing of $L$ with an interior point of a side of $P$, then that point is the endpoint of only one segment of $L$, and is the only point of intersection of $L$ with that side. If the segment endpoint is a vertex of $P$, then it may be the endpoint of of two segments of $L$, but in that case it is the only point of intersection of $L$ with both sides incident to the vertex. So in either case each endpoint of a segment of $L$ uses up at least one side of $P$. As $P$ has $p$ sides, the number of segments is at most $\lfloor p/2\rfloor$.
\end{proof}

\begin{corollary}
\label{cor:few-segs}
Let $\mathcal{A}$ be an arrangement of $\ell$ lines, and let $P$ be a simple $p$-gon whose interior is disjoint from the crossings of $\mathcal{A}$. Then the lines of $\mathcal{A}$ intersect the interior of $P$ in at most $\ell\cdot\lfloor p/2\rfloor$ open line segments.
\end{corollary}

\begin{lemma}
\label{lem:region-tree}
Let $\mathcal{A}$ be an arrangement of lines and $P$ be a simple polygon that does not contain any crossing point of $\mathcal{A}$.
Then the lines of $\mathcal{A}$ partition the interior of $P$ into regions in such a way that the graph of regions and their adjacencies forms a tree.
\end{lemma}

\begin{proof}
To show that the graph of regions and adjacencies is connected, consider any two regions $R_i$ and $R_j$, and choose a curve $C$ within the interior of $P$ connecting any point in $R_i$ to any point in $R_j$. Then the sequence of regions crossed by $C$ forms a walk in the region adjacency graph connecting $R_i$ to $R_j$.

To show that the graph of regions and adjacencies has no simple cycle, assume for a contradiction that there is such a cycle. Then by choosing a representative point within each region of the cycle, and connecting these points by curves that pass between adjacent regions without crossing any other regions, we can form a simple closed curve $C$ in the plane that crosses the lines of $\mathcal{A}$ in exactly the order given by the cycle. By the Jordan curve theorem, each line that crosses into the interior of $C$ must cross out of $C$ at another point. Two crossings of $C$ by the same line cannot be adjacent in the cyclic order of crossings, for then the graph cycle corresponding to $C$ would not be simple. Therefore, $C$ is crossed by at least two lines, in alternating order. But this can happen only when $C$ contains the crossing point of these lines, an impossibility as $C$ is entirely contained in $P$ which we assumed to enclose no crossings. This contradiction shows that a simple cycle does not exist.

As a connected graph with no simple cycles, the graph of regions and adjacencies must be a tree.
\end{proof}

\begin{lemma}
\label{lem:use-two}
Let $\mathcal{A}$ be an arrangement of lines and $P$ be a simple polygon that does not contain any crossing point of $\mathcal{A}$. Let $\mathcal{S}$ be the system of disjoint open line segments formed by intersecting the lines of $\mathcal{A}$ with the interior of $P$, and let $Q$ be another simple polygon, disjoint from $P$, such that each vertex of $Q$ lies on a segment of $\mathcal{S}$.
Then at least two segments of $\mathcal{S}$ are disjoint from the interior of~$Q$.
\end{lemma}

\begin{proof}
Because the graph of regions and adjacencies formed in $P$ by $\mathcal{A}$ is a tree (\autoref{lem:region-tree}), it has at least two leaves. Let $s$ be either of the two segments of $\mathcal{S}$ separating one of these leaf regions from the rest of $P$. Then no edge of $Q$ can enter the interior of this leaf region, because there is no other segment available to be the endpoint of this edge. Therefore, $Q$ remains entirely on one side of $s$, and $s$ is disjoint from the interior of $Q$. As there were at least two choices for $s$, there are at least two segments of $\mathcal{S}$ that are disjoint from the interior of~$Q$.
\end{proof}

Putting these observations together, we have:

\begin{lemma}
\label{lem:deep-nest}
Let $\mathcal{A}$ be an arrangement of $\ell$ lines, let $p$ and $r$ be positive integers, and
suppose that $2(r-1)>\ell\cdot\lfloor p/2\rfloor$.
Then it is not possible to draw a $(p,r)$-nest in such a way that the polygon vertices of the nest and its egg all lie on lines of $\mathcal{A}$.
\end{lemma}

\begin{proof}
Suppose for a contradiction that we have drawn a $(p,r)$-nest with all points on lines of $\mathcal{A}$. Let $\mathcal{S}$ be the system of disjoint open line segments formed by intersecting the lines of $\mathcal{A}$ with the outer polygon of the nest. Then $|\mathcal{S}|\le \ell\cdot\lfloor p/2\rfloor$ by \autoref{cor:few-segs}, and each of the $r-1$ remaining polygons of the nest use up at least two of the segments of $\mathcal{S}$ by \autoref{lem:use-two}. Therefore, if $2(r-1)>\ell\cdot\lfloor p/2\rfloor$ (as we supposed in the statement of the lemma), there will be no segments remaining for the egg to lie on. Therefore, a drawing meeting these conditions is impossible.
\end{proof}

\begin{figure}[t]
\centering\includegraphics[width=\textwidth]{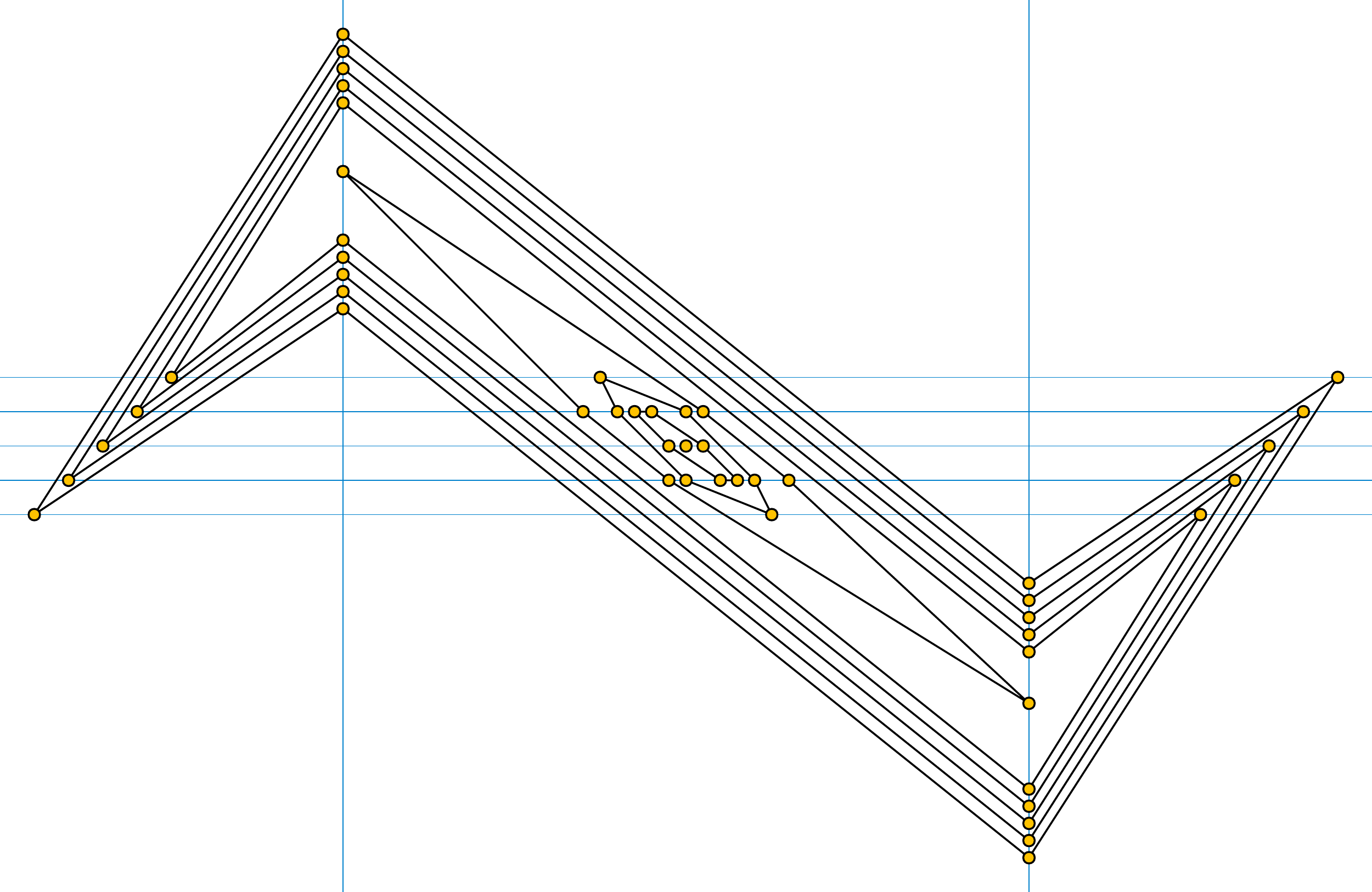}
\caption{An arrangement of $\ell$ lines can support $\lfloor 3(\ell-1)/2-1 \rfloor$ nested hexagons surrounding a central point, with the point and the hexagon vertices all on the lines, and all arrangement crossings exterior to all hexagons.}
\label{fig:hexnest}
\end{figure}

In the particular case of a $(6,r)$-nest (as used in \autoref{fig:spiderwebs}), \autoref{lem:deep-nest} states that a drawing that does not contain an arrangement crossing cannot exist for $r\ge 3\ell/2+2$.
This is close to tight: \autoref{fig:hexnest} shows how to draw a $(6,r)$-nest with all polygon vertices and the egg on $\ell$ lines, for $r=3\ell/2-O(1)$.

\subsection{The main result}

\begin{proof}[Proof of \autoref{thm:cubic}]
We wish to show the existence of a planar cubic bipartite graph that cannot be drawn on $\ell$ lines, for a given parameter $\ell$. Consider a cubic bipartite 3-vertex-connected planar graph formed, as in \autoref{fig:spiderwebs}, by at least $\tbinom{l}{2}+2$ subunits, each of which must be drawn as a $(6,r)$-nest, for $r=\lceil 3\ell/2\rceil+2$. Because there are $O(\ell^2)$ subunits, each of size $O(\ell)$, and $O(\ell)$ vertices surrounding the subunits, the total size of the resulting graph is $O(\ell^3)$.

To argue that the resulting graph cannot be drawn on $\ell$ lines, we consider an arbitrary arrangement $\mathcal{A}$ with $\ell$ lines, and prove that the graph cannot be drawn with all of its vertices on $\mathcal{A}$. Among the graph's $\tbinom{l}{2}+2$ subunits, one subunit (the one from which the outer face was chosen) can surround all the others, but the rest must be drawn in disjoint regions of the plane. Because there are more remaining subunits than the number of crossing points of $\mathcal{A}$ (which is at most $\tbinom{\ell}{2}$), at least one subunit must be drawn in such a way that it does not contain any of the crossing points of $\mathcal{A}$. However, by \autoref{lem:deep-nest}, this is impossible.
\end{proof}

\begin{figure}[t]
\centering\includegraphics[width=0.8\textwidth]{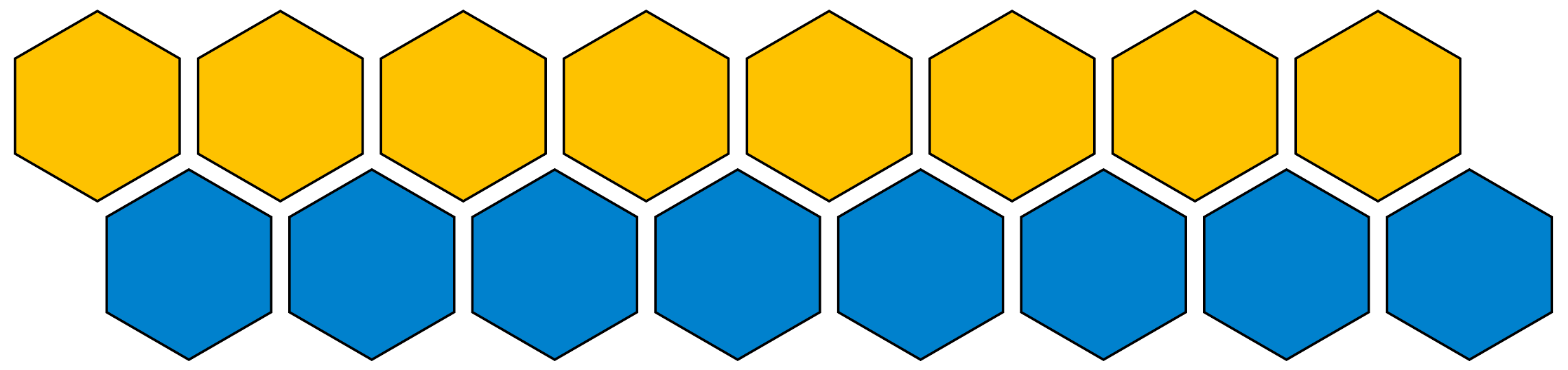}
\caption{Zigzag pattern of subunits (shown schematically as hexagons) used to construct a graph of  pathwidth $O(1)$ that cannot be drawn on few lines.}
\label{fig:zigzag}
\end{figure}
When the subunits of the graph are arranged as in \autoref{fig:spiderwebs}, into a compact hexagonal grid in the plane, then (as the figure shows) the vertices of the whole graph can be covered by $O(\ell)$ lines, even though we have proved that there is no cover by $\ell$ lines.
Alternatively, it is possible to arrange the subunits into a linear zig-zag pattern (\autoref{fig:zigzag}), preserving the 3-vertex-connectedness of the resulting graph and still requiring only $O(\ell^2)$ vertices to surround the subunits. When the subunits are arranged in this way, the resulting graph might require a larger number of lines to cover its vertices (i.e., we do not have tight bounds on the number of lines needed for a graph of this form), but it has pathwidth $O(1)$.

\section{Series-parallel graphs}

A \emph{two-terminal series-parallel graph} (series-parallel graph, for short) is a graph with two distinct designated \emph{terminal} vertices $s$ and $t$ formed recursively from smaller graphs of the same type (starting from a single edge) by two operations:
\begin{itemize}
\item Series composition: given two series-parallel graphs $G_1$ and $G_2$ with terminals $s_1$, $t_1$, $s_2$, and $t_2$, form their disjoint union, and then merge vertices $s_2$ and $t_1$ into a single vertex. Let the terminals of the resulting merged graph be the unmerged terminals of the given graphs, $s_1$ and $t_2$. Series composition forms an associative binary operation on these graphs (if we perform series compositions on a sequence of more than two graphs, the order in which we perform the compositions does not affect the result).
\item Parallel composition: given two series-parallel graphs $G_1$ and $G_2$ with terminals $s_1$, $t_1$, $s_2$, and $t_2$, form their disjoint union,  merge vertices $s_1$ and $s_2$ into a single vertex, and similarly merge $t_1$ and $t_2$ into a single vertex. Let the terminals of the resulting merged graph be the resulting merged vertices. Parallel composition forms an associative and commutative binary operation on these graphs (if we perform series compositions on a set of more than two graphs, neither the order of the two graphs in each composition nor the order in which we perform the compositions affects the result).
\end{itemize}
These graphs have treewidth two, and every graph of treewidth two is a subgraph of a series-parallel graph. They are automatically planar, and they include every outerplanar graph.

\begin{figure}[t]
\centering\includegraphics[width=\textwidth]{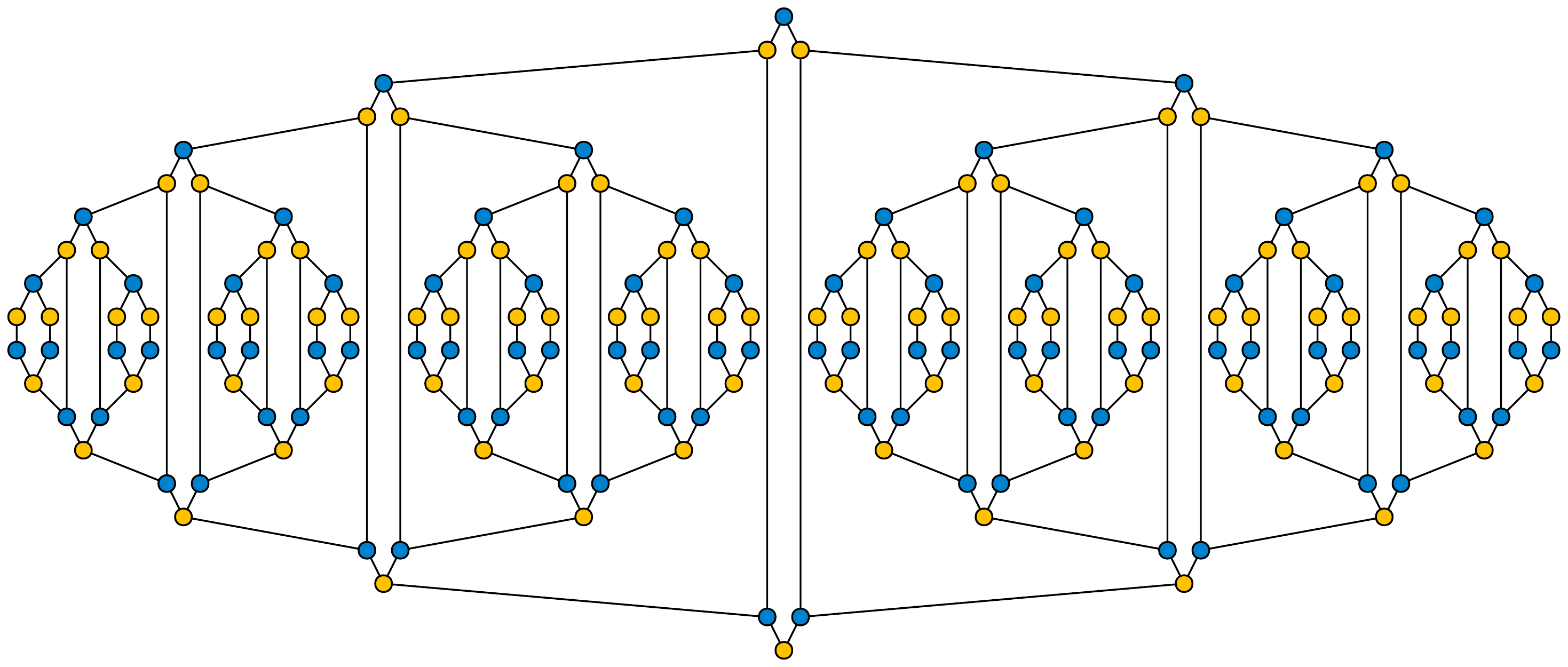}
\caption{A subcubic bipartite series-parallel graph that cannot be drawn on few lines.}
\label{fig:serpar}
\end{figure}

\begin{figure}[t]
\centering\includegraphics[scale=0.5]{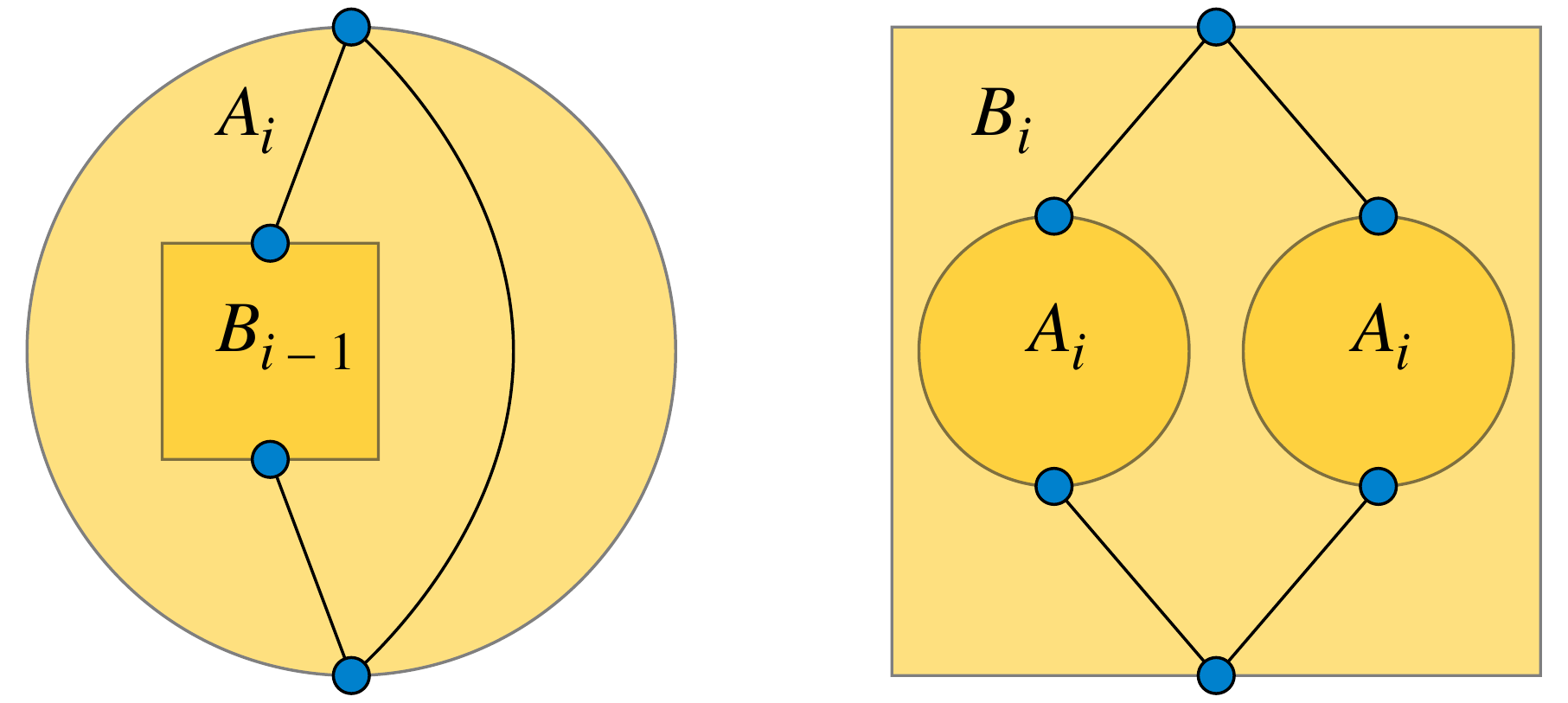}
\caption{Recursive construction of $A_i$ (left) and $B_i$ (right)}
\label{fig:recurse}
\end{figure}

We will recursively construct two families  of series-parallel graphs
$A_i$ and $B_i$ that cannot be drawn on a bounded number of lines. \autoref{fig:serpar} shows the graph $B_5$ from this family.
To construct these graphs, let $A_1$ be a series-parallel graph with one edge and two terminal vertices. Then:
\begin{itemize}
\item For each $i\ge 1$, let $B_i$ be the graph formed as the parallel composition of two subgraphs,
each of which is the series composition of an edge, $A_i$, and another edge (\autoref{fig:recurse}, right).
\item For each $i>1$, let $A_i$ be the graph formed as the parallel composition of two subgraphs,
one of which is a single edge and the other of which is the series composition of an edge, $B_{i-1}$, and another edge (\autoref{fig:recurse}, left).
\end{itemize}

It follows by induction that these graphs are subcubic, with degree two at their two terminals, and that they are bipartite, with a 2-coloring (shown in the figure) in which the two terminals have different colors. In the figure, the upper blue and lower yellow vertices are the terminals of a graph $B_i$ at some level of the construction, while the upper yellow and lower blue vertices are the terminals of a graph $A_i$ at some level of the construction.

Because they are 2-vertex-connected but not 3-vertex-connected, these graphs have many planar embeddings. 
The planar embeddings of any 2-connected graph may be understood in terms of its SPQR tree~\cite{Mac-DMJ-37},
which in the case of a series-parallel graph is more or less the same as the expression tree of series and parallel compositions from which it was formed (associating consecutive compositions of the same type into a single multi-operand operation). Because of this equivalence, the embeddings of the graphs $A_i$ and $B_i$ may be generated from the embedding shown in the figure by two types of change: any collection of subgraphs connecting two opposite terminals of the graph may be flipped, giving a mirror-image embedding of that subgraph within the larger graph, and any face of the resulting embedding may be chosen as the outer face.
We will show that these embeddings always contain large nested sets of hexagons; this property forms the basis for our argument about drawings on few lines.

\begin{figure}[t]
\centering\includegraphics[scale=0.5]{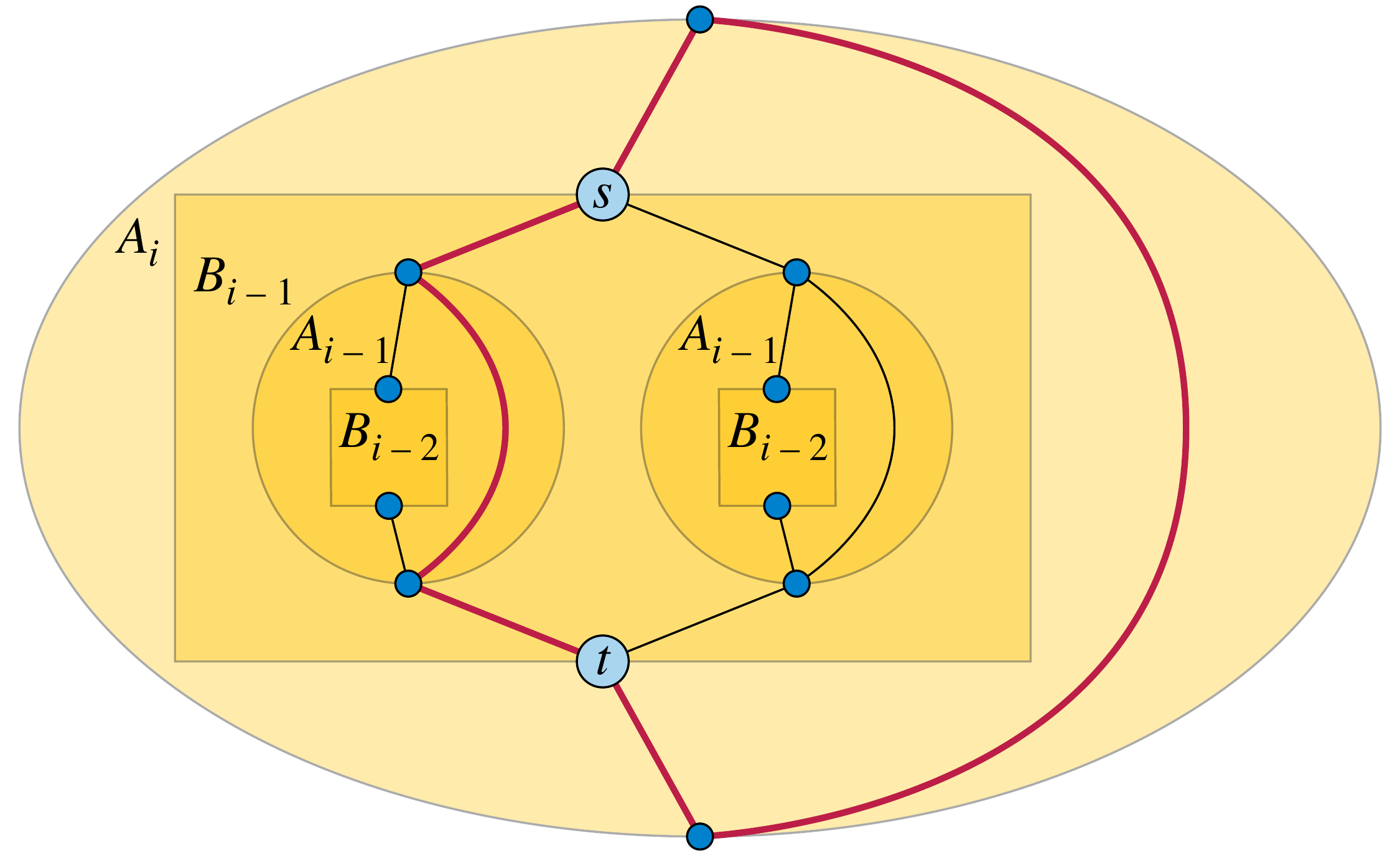}
\caption{Illustration for \autoref{lem:sp-nest}. No matter how $A_i$ is embedded, it will contain a 6-cycle (red) surrounding one of the copies of $A_{i-1}$ from which it is formed.}
\label{fig:surrounded}
\end{figure}

\begin{lemma}
\label{lem:sp-nest}
Every planar embedding of $A_i$ in which the two terminals belong to the outer face contains a $(6,i-1)$-nest.
\end{lemma}

\begin{proof}
For $i=1$, a $(6,i-1)$-nest consists of a single point (the egg), and the result follows trivially.
Otherwise, let $s$ and $t$ be the two terminals of the $B_{i-1}$ subgraph from which the given $A_i$ graph is formed (\autoref{fig:surrounded}). Then, in the $A_i$ graph, there are three length-three paths from $s$ to $t$:
one through the two terminals of the $A_i$ graph,
and one through each of the two $A_{i-1}$ subgraphs from which the $B_{i-1}$ subgraph is formed.

By the assumption that the two terminals of $A_i$ belong to the outer face, 
there is a six-vertex cycle combining two of these three length-three paths, one through the two terminals of $A_i$ and one through one of the copies of $A_{i-1}$, that surrounds the other copy of $A_{i-1}$. By induction, this surrounded copy contains a $(6,i-2)$-nest, which together with the cycle that surrounds it forms a $(6,i-1)$-nest.
\end{proof}

\begin{lemma}
\label{lem:sp-many-nests}
Let $j\le i$ be two positive integers. Then every planar embedding of $B_i$ contains at least $2^j-1$ disjointly-embedded $(6,i-j)$-nests.
\end{lemma}

\begin{proof}
$B_j$ is recursively constructed from $2^j$ copies of $A_{i-j+1}$.
At most one of these copies can contain the outer face of the embedding, so at least $2^j-1$ copies
are embedded with their two terminals outermost. The result follows from \autoref{lem:sp-nest}.
\end{proof}

\begin{proof}[Proof of \autoref{thm:treelike}, series-parallel case]
The theorem claims that there exists a cubic bipartite series-parallel graph that cannot be drawn on $\ell$ lines. To prove this, choose $j$ such that $2^j\ge \tbinom{\ell}{2}+2$
and choose $i$ such that $i-j\ge 3\ell/2+2$. We claim that, for this case, $B_i$ has the required properties. By \autoref{lem:sp-many-nests}, it contains at least $\tbinom{\ell}{2}+1$ disjointly embedded $(6,i-j)$-nests, enough to ensure that at least one of them does not contain any crossing points of any given arrangement of $\ell$ lines.  By \autoref{lem:deep-nest}, a nest of this depth that does not contain any crossing points cannot be drawn with its vertices on $\ell$ lines.
\end{proof}

\section{Apex-tree graphs}

\subsection{Apex-tree graphs requiring many lines}
\autoref{fig:apex-tree} depicts a graph in the form of a tree (blue and yellow vertices) plus one additional vertex (red); such a graph has been called an \emph{apex-tree}, and the additional vertex is the \emph{apex}
The tree in the figure can be constructed from the series-parallel graph $B_4$ of the previous section by contracting half of the vertices (one vertex from each pair of terminals) into a single supervertex. Alternatively, it can be constructed from a complete binary tree by subdividing every non-leaf edge and then connecting each subdivision vertex and each leaf vertex to the apex.
These graphs are subcubic except at the apex, and bipartite.

\begin{figure}[t]
\centering\includegraphics[width=0.8\textwidth]{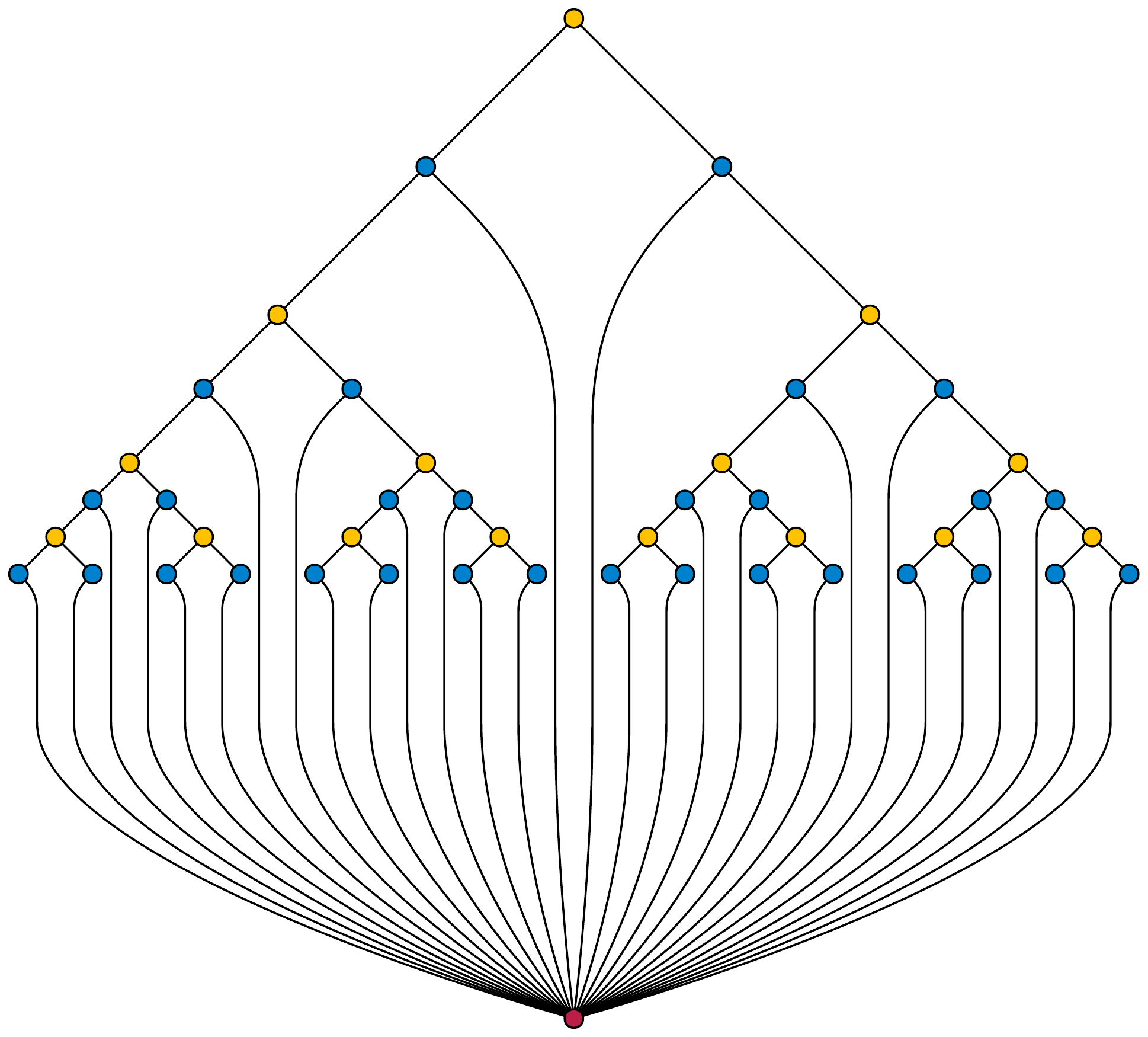}
\caption{An apex-tree graph, subcubic except at the apex, that cannot be drawn on few lines.}
\label{fig:apex-tree}
\end{figure}

As with the earlier series-parallel graphs, we will prove that these graphs cannot be drawn on a sublogarithmic number of lines. An obstacle to the proof, however, is that they contain no  $(p,r)$-nests for $r>1$, nor can any such nest exist in any apex-tree. The reason is that, in an apex-tree, all cycles contain the apex. Therefore, there can be no two disjoint cycles, and no nests of two or more disjoint cycles. Nevertheless, these graphs do contain nest-like structures. We define a \emph{$(p,r)$-near nest} in an embedded plane graph to be a collection of $p$ $r$-cycles, plus one additional vertex (the egg), such that all cycles contain the egg in their interior, the cycles are edge-disjoint, and any two of them share at most one vertex with each other. Then the following lemma is an analogue of \autoref{lem:deep-nest} for near-nests:

\begin{lemma}
\label{lem:deep-near-nest}
Let $\mathcal{A}$ be an arrangement of $\ell$ lines, let $p$ and $r$ be positive integers, and
suppose that $r-1>\ell\cdot\lfloor p/2\rfloor$.
Then it is not possible to draw a $(p,r)$-near-nest in such a way that the polygon vertices of the nest and its egg all lie on lines of $\mathcal{A}$.
\end{lemma}

\begin{proof}
Suppose for a contradiction that we have drawn a $(p,r)$-near-nest with all points on lines of $\mathcal{A}$. Let $\mathcal{S}$ be the system of disjoint open line segments formed by intersecting the lines of $\mathcal{A}$ with the outer polygon of the nest. Then $|S|\le \ell\cdot\lfloor p/2\rfloor$ by \autoref{cor:few-segs}, and each of the $r-1$ remaining polygons of the nest use up at least one of the segments of $\mathcal{S}$ by \autoref{lem:use-two} and by the fact that at most one of the two extreme segments of the polygon can be shared with other polygons interior to it. Therefore, if $r-1>\ell\cdot\lfloor p/2\rfloor$ (as we supposed in the statement of the lemma), there will be no segments remaining for the egg to lie on. Therefore, a drawing meeting these conditions is impossible.
\end{proof}

Analogously to \autoref{lem:sp-nest} and \autoref{lem:sp-many-nests}, we have:

\begin{lemma}
\label{lem:at-many-nests}
For every planar embedding of a graph like the one of \autoref{fig:apex-tree} formed from a complete binary tree of height $i$, and for every $j\le i$,
the embedding contains $2^j-1$ disjoint $(4,i-j)$-near-nests.
\end{lemma}

\begin{proof}
This result follows immediately from \autoref{lem:sp-many-nests},
which proves the existence of a $(6,i-j)$-nest in the corresponding series-parallel graphs,
together with the observations that every planar embedding of our apex-trees
can be expanded to a planar embedding of the corresponding series-parallel graphs,
and that every 6-cycle of a $(6,i-j)$-nest in the expanded series-parallel graph
has three of its vertices contracted into the apex of the apex-tree graph.

Alternatively, one could prove the result by repeating the proof of \autoref{lem:sp-many-nests} with minor modifications.
\end{proof}

\begin{proof}[Proof of \autoref{thm:treelike}, apex-tree case]
The theorem claims that there exists a bipartite apex-tree graph, subcubic except at its apex, that cannot be drawn on $\ell$ lines. To prove this, choose $j$ such that $2^j\ge \tbinom{\ell}{2}+2$
and choose $i$ such that $i-j-1> 2\ell$. Form an apex-tree graph as above from a complete binary tree of height~$i$.
We claim that, for this case, the resulting apex-tree graph has the required properties. For, by \autoref{lem:at-many-nests}, it contains at least $\tbinom{\ell}{2}+1$ disjointly embedded $(4,i-j)$-near-nests, enough to ensure that at least one of them does not contain any crossing points of any given arrangement of $\ell$-lines.  By \autoref{lem:deep-near-nest}, a nest of this depth that does not contain any crossing points of the $\ell$ lines cannot be drawn with its vertices on the lines.
\end{proof}

\subsection{Drawing apex-tree graphs on few lines}

Recall that \autoref{thm:apex-tree-draw} states that we can draw any apex-tree graph planarly on $O(\log n)$ parallel lines. To do so, we adapt a standard tool from tree drawing, the \emph{heavy path decomposition}~\cite{SleTar-JCSS-83}, to draw any tree with its vertices on the points of a grid of height $\log_2 n$ and width $n$ (in particular, on $O(\log n)$ horizontal lines) in such a way that the resulting drawing can be extended to a drawing of an apex-tree, by adding one more vertex adjacent to any subset of the tree vertices.

The heavy path decomposition of a tree is obtained by choosing one \emph{heavy edge} for each non-leaf vertex of the tree, an edge connecting it to the subtree with the largest number of vertices (breaking ties arbitrarily). The connected components of the subgraph formed by the heavy edges are \emph{heavy paths}, including as a special case length-zero paths for leaf vertices that were not chosen by their parent. The heavy paths partition the vertices of the tree. By induction, a vertex $v$
that can reach a leaf by a path that includes $i$ non-heavy edges must be the root of a subtree containing at least $2^{i+1}-1$ vertices (including $v$ itself). Therefore, in a tree with $n>1$ vertices,
every root-to-leaf path contains at most $\log_2 n-1$ non-heavy edges.

\begin{figure}[t]
\centering\includegraphics[width=0.5\textwidth]{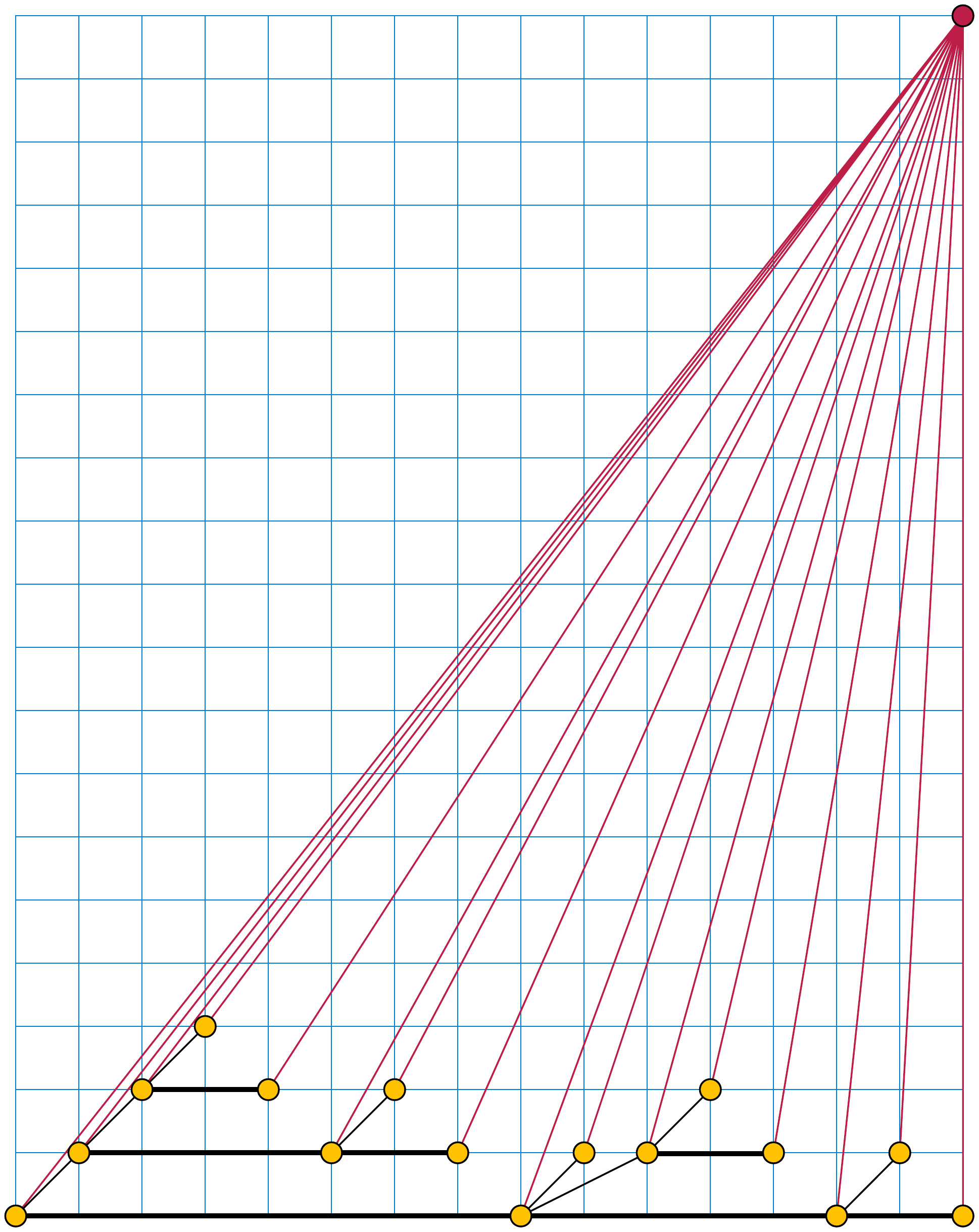}
\caption{Drawing an apex-tree on a grid. The thick horizontal black lines depict the heavy path decomposition of the given tree. Note that although the grid size is approximately $(n+\log_2 n)\times n$, only the bottom $\log_2 n$ horizontal grid lines and the top horizontal grid line are occupied by vertices.}
\label{fig:apex-tree-grid}
\end{figure}

\begin{proof}[Proof of \autoref{thm:apex-tree-draw}]
To draw the given tree on a grid, we traverse the tree in preorder, ordering the children at each vertex so that the heavy edge is last. We let the $x$-coordinate of each vertex be its position in this preorder listing, and we let the $y$-coordinate be the number of non-heavy edges on the path from the vertex to the root. These choices give unique coordinates for each vertex on a grid of height $\log_2 n$ and width $n$, as claimed.
Each tree edge either connects two consecutive vertices on the same level of the grid (on the same heavy path), or it connects vertices on consecutive levels (a parent and child not connected by a heavy edge)
whose $x$-coordinates both are less than the next vertex on the same level as the parent.
Therefore, the drawing has no crossings.

All edges of this tree drawing have slope in the interval $[0,1]$. The traversal ordering ensures that, for each vertex $v$, and each vertex $w$ with a higher $y$-coordinate than $v$, one of the following is true: $w$ has smaller $x$-coordinate than $v$, $w$ is a descendant of $v$, or $w$ is a descendant of a vertex that is placed below and to the right of~$v$. In all three cases, neither $w$ nor any edge incident to $w$ can block the visibility from $v$ upwards and to the right through lines of slope greater than one.
Therefore, if we place the apex $n+1$ units above the upper right corner of the grid, it will be visible to all tree vertices by unobstructed lines of sight and we can complete the drawing of any  apex-tree consisting of the given tree and one apex.
\end{proof}

The construction is depicted in \autoref{fig:apex-tree-grid}.

\section{Conclusions and open problems}

We have found planar 3-regular bipartite graphs of size cubic in $\ell$ that cannot be drawn on $\ell$ lines, cubic bipartite series-parallel graphs of size exponential in $\ell$ that cannot be drawn on $\ell$ lines, and apex-trees of size exponential in $\ell$ that cannot be drawn on $\ell$ lines. For apex-trees the exponential size bound is necessary, although there may still be room for tightening the gap between the exponential upper and lower bounds. For the other two classes of graphs, we do not know whether our results are tight.  Stefan Felsner and Alexander Wolff have recently proven that every 4-vertex-connected maximal planar graph of size at most quadratic in $\ell$ may be drawn on $\ell$ pseudolines, and that it is NP-hard to find drawings on two lines~\cite{BieEvaFel-ms-19}.
Do there exist planar graphs of subcubic size that cannot be drawn on $\ell$ lines? Do there exist series-parallel graphs of polynomial size that cannot be drawn on $\ell$ lines? How well can the optimal number of lines be approximated? We leave these problems as open for future research.

\bibliographystyle{plainurl}
\bibliography{many-lines}
\end{document}